\documentclass[twoside,onecolumn]{article}

\usepackage[english]{babel} 

\usepackage[hmarginratio=1:1,top=32mm,columnsep=20pt]{geometry} 
\usepackage[hang, small,labelfont=bf,up,textfont=it,up]{caption} 

\usepackage{geometry} 
\usepackage{graphicx} 

\usepackage{booktabs} 
\usepackage{array} 
\usepackage{paralist} 
\usepackage{verbatim} 
\usepackage{subfig} 
\usepackage{amsmath}	
\numberwithin{equation}{section}	
\usepackage{amsfonts}	
\usepackage{amsthm}	
\usepackage{cases}
\usepackage{bbm}		
\usepackage{braket}	
\usepackage{mathtools}	
\usepackage[english]{varioref}	
\usepackage{float}		

\usepackage{fancyhdr} 
\theoremstyle{plain}
\newtheorem{Example}{Example}[section]
\theoremstyle{plain}
\newtheorem{Remark}{Remark}[section]
\theoremstyle{plain}
\newtheorem{Proposition}{Proposition}[section]
\theoremstyle{plain}
\newtheorem{Lemma}{Lemma}[section]
\theoremstyle{plain}
\newtheorem{Theorem}{Theorem}[section]
\theoremstyle{plain}

\theoremstyle{plain}
\newtheorem{Definition}{Definition}[section]
\theoremstyle{plain}
\newtheorem{assumption}{Assumption}[section]

\DeclarePairedDelimiter{\abs}{\lvert}{\rvert}

\DeclareMathOperator*{\var}{\text{var}}
\DeclareMathOperator*{\essinf}{ess\,inf}

\usepackage{sectsty}
\allsectionsfont{\sffamily\mdseries\upshape} 

\usepackage{abstract} 

\usepackage{titlesec} 
\titleformat{\section}[block]{\large\bfseries\centering}{\thesection.}{1em}{} 
\titleformat{\subsection}[block]{\large\bfseries}{\thesubsection.}{1em}{} 

\usepackage[nottoc,notlof,notlot]{tocbibind} 
\usepackage[titles,subfigure]{tocloft} 

\usepackage[a-1b]{pdfx} 
\usepackage{titling} 
\usepackage{hyperref}

\hypersetup{%
    pdfpagemode={UseOutlines},
    bookmarksopen,
    pdfstartview={FitH},
    colorlinks,
    linkcolor={black},
    citecolor={black},
    urlcolor={black}
  }

\pretitle{\begin{center}\Huge\bfseries} 
\posttitle{\end{center}} 
\title{Optimal excess-of-loss reinsurance for stochastic factor risk models} 
\author{%
\textsc{Brachetta M.}\thanks{Department of Economics, University of Chieti-Pescara, Italy.} \\[1ex]
\normalsize \href{mailto:matteo.brachetta@unich.it}{matteo.brachetta@unich.it}
\and
\textsc{Ceci, C.}\footnotemark[1]\\[1ex]
\normalsize \href{mailto:c.ceci@unich.it}{c.ceci@unich.it}
}
\date{} 

\providecommand{\keywords}[1]{\textbf{\textit{Keywords:}} #1}
\providecommand{\jelcodes}[1]{\textbf{\textit{JEL Classification codes:}} #1}
\providecommand{\msccodes}[1]{\textbf{\textit{MSC Classification codes:}} #1}

\begin{document}
\maketitle

\begin{abstract}
\noindent
We study the optimal excess-of-loss reinsurance problem when both
the intensity of the claims arrival process and the claim size distribution
are influenced by an exogenous stochastic factor. We assume
that the insurer's surplus is governed by a marked point process
with dual-predictable projection affected by an environmental factor
and that the insurance company can borrow and invest money at a
constant real-valued risk-free interest rate $r$. Our model allows for stochastic
risk premia, which take into account risk fluctuations. Using stochastic
control theory based on the Hamilton-Jacobi-Bellman equation, we analyze
the optimal reinsurance strategy under the criterion
of maximizing the expected exponential utility of the terminal wealth.
A verification theorem for the value function in terms of classical solutions
of a backward partial differential equation is provided. Finally, some numerical results are discussed.
\end{abstract}

\noindent\keywords{optimal reinsurance, excess-of-loss reinsurance, Hamilton-Jacobi-Bellman equation, stochastic factor model, stochastic control.}\\
\noindent\jelcodes{G220, C610.}\\
\noindent\msccodes{93E20, 91B30, 60G57, 60J75.}\\


\section{Introduction}

In this paper we analyze the optimal excess-of-loss reinsurance problem from the insurer's point of view, under the criterion of maximizing the expected utility of the terminal wealth. It is well known that the reinsurance policies are very effective tools for risk management. In fact, by means of a risk sharing agreement, they allow the insurer to reduce unexpected losses, to stabilize operating results, to increase business capacity and so on. Among the most common arrangements, the proportional and the excess-of-loss contracts are of great interest. The former was intensively studied in \cite{irgens_paulsen:optcontrol}, \cite{liuma:optreins}, \cite{liangetal:optreins}, \cite{liangbayraktar:optreins}, \cite{zhuetal:reins_defaultable}, \cite{BC:IME2018} and references therein. The latter was investigated in these articles: in \cite{zhang_zhou_guo:optcomb} and \cite{meng_zhang_2010}, the authors proved the optimality of the excess-of-loss policy under the criterion of minimizing the ruin probability, with the surplus process described by a Brownian motion with drift; in \cite{zhaoetal:2013} the Cram\'er-Lundberg model is used for the surplus process, with the possibility of investing in a financial market represented by the Heston model; in \cite{shengrongzhao:optreins} and \cite{liguo:2013} the risky asset is described by a Constant Elasticity of Variance (CEV) model, while the surplus is modelled by the Cram\'er-Lundberg model and its diffusion approximation, respectively; finally, in \cite{scandinavian:2018} the authors studied a robust optimal strategy under the diffusion approximation of the surplus process.

The common ground of the cited works is the underlying risk model, which is the Cram\'er-Lundberg model (or its diffusion approximation)%
\footnote{See \cite{lundberg:1903}, \cite{schmidli:2018risk}.}.
In the actuarial literature it is of great importance, because it is simple enough to perform calculations. In fact, the claims arrival process is described by a Poisson process with constant intensity (or a Brownian motion, in the diffusion model). Nevertheless, as noticed by many authors (e.g. \cite{grandell:risk}, \cite{hipp:stoch.control_applications}),
it needs generalization in order to take into account the so callled \textit{size fluctuations} and \textit{risk fluctuations}, i.e. variations of the number of policyholders and modifications of the underlying risk, respectively.

The main goal of our work is to extend the classical risk model by modelling the claims arrival process as a marked point process with dual-predictable projection affected by an exogenous stochastic process $Y$. More precisely, both the intensity of the claims arrival process and the claim size distribution are influenced by $Y$. Thanks to this environmental factor, we achieve a reasonably realistic description of any risk movement. For example, in automobile insurance $Y$ may describe weather conditions, road conditions, traffic volume and so on. All these factors usually influence the accident probability as well as the damage size.

Some noteworthy attempts in that direction can be found in \cite{liangbayraktar:optreins} and \cite{BC:IME2018}, where the authors studied the optimal proportional reinsurance. In the former, the authors considered a Markov-modulated compound Poisson process, with the (unobservable) stochastic factor described by a finite state Markov chain. In the latter, the stochastic factor follows a general diffusion. In addition, in \cite{BC:IME2018} the insurance and the reinsurance premia are not evaluated by premium calculation principles (see \cite{young:premium_princ}), because they are stochastic processes depending on $Y$. In our paper, we extend further the risk model, because the claim size distribution is influenced by the stochastic factor, which is described by a diffusion-type stochastic differential equation (SDE). In addition, we study a different reinsurance contract, which is the excess-of-loss agreement.

In our model the insurer is also allowed to lend or borrow money at a given interest rate $r$. During the last years, negative interest rates drew the attention of many authors. For example, since June 2016 the European Central Bank (ECB) fixed a negative Deposit facility rate, which is the interest banks receive for depositing money within the ECB overnight. Nowadays, it is $-0.4\%$. As a consequence, in our framework $r\in\mathbb{R}$. We point out that there is no loss of generality due to the absence of a risky asset, because as long as the insurance and the financial markets are independent (which is a standard hypothesis in non-life insurance), the optimal reinsurance strategy turns out to depend only on the risk-free asset (see \cite{BC:IME2018} and references therein). As a consequence, the optimal investment strategy can be eventually obtained using existing results in the literature. 

The paper is organized as follows: in Section~\ref{section:model}, we formulate the model assumptions and describe the maximization problem; in Section~\ref{section:HJB} we derive the Hamilton-Jacobi-Bellman (HJB) equation; in Section~\ref{section:reinsurance}, we investigate the candidate optimal strategy, which is suggested by the HJB derivation; in Section~\ref{section:verification}, we provide the verification argument with a probabilistic representation of the value function; finally, in Section~\ref{section:sim} we perform some numerical simulations. 


\section{Model formulation}
\label{section:model}

Let $(\Omega,\mathcal{F},\mathbb{P},\{\mathcal{F}_t\}_{t\in[0,T]})$ be a complete probability space endowed with a filtration which satisfies the usual conditions, where $T>0$ is the insurer's time horizon. We model the insurance losses through a marked point process $\{ (T_n, Z_n)\} _{n\geq 1}$ with local characteristics   influenced by an environment stochastic factor $Y\doteq\{Y_t\}_{t\in[0,T]}$.  Here, the sequence $\{ T_n \}_{n\geq 1}$ describes the claim arrival process and  $\{ Z_n\} _{n\geq 1}$  the corresponding claim sizes.
Precisely, $T_n$, $n=1, \dots$, are $\{\mathcal{F}_t\}_{t\in[0,T]}$ stopping times such that $T_n < T_{n+1}$ a.s. and $Z_n$, $n=1, \dots$, are $(0,+\infty)$-random variables such that $\forall n=1, \dots$,  $Z_n$ is $\mathcal{F}_{T_n}$-measurable.

The stochastic factor  $Y$ is defined as  the unique strong solution to the following SDE:
\begin{equation}
\label{eqn:stochasticfactor}
dY_t = b(t,Y_t)\,dt + \gamma(t,Y_t)\,dW^{(Y)}_t , \qquad Y_0\in\mathbb{R} ,
\end{equation}
where $\{W^{(Y)}_t\}_{t\in [0,T]}$ is a standard Brownian motion on $(\Omega,\mathcal{F},\mathbb{P},\{\mathcal{F}_t\}_{t\in[0,T]})$. We assume that the following conditions hold true:
\begin{gather}
\label{eqn:solutionY}
\mathbb{E}\biggl[\int_0^T\abs{b(t,Y_t)}\,dt+\int_0^T \gamma(t,Y_t)^2\,dt\biggr]<\infty , \\
\label{eqn:solutionY2}
\sup_{t\in[0,T]}{\mathbb{E}[\abs{Y_t}^2]}<\infty .
\end{gather}
We will denote by  $\{\mathcal{F}^Y_t\}_{t\in[0,T]}$ the natural filtration generated by the process $Y$.

The random measure corresponding to the losses process $\{ (T_n, Z_n)\} _{n\geq 1}$  is given by
\begin{equation}
\label{eqn:random_measure}
m(dt,dz) \doteq \sum_{n\ge1}\delta_{(T_n,Z_n) (dt,dz)}\mathbbm{1}_{\{T_n\le T\}} ,
\end{equation}
where $\delta_{(t,x)}$ denotes the Dirac measure located at point $(t,x)$. We assume that its $\{\mathcal{F}_t\}_{t\in[0,T]}$-dual predictable projection $\nu(dt,dz)$ has the form
\begin{equation}
\label{eqn:dual_projection}
\nu(dt,dz) = dF(z,Y_t)\lambda(t,Y_t)\,dt ,
\end{equation}
where
\begin{itemize}
\item $F(z,y):[0,+\infty)\times\mathbb{R}\to[0,1]$ is such that  $\forall y\in\mathbb{R}$, $F(\cdot,y)$ is a distribution function, with $F(0,y)=0$;
\item $\lambda(t,y):[0,T]\times\mathbb{R}\to(0,+\infty)$ is a strictly positive measurable function.
\end{itemize}

In the sequel, we will assume the following integrability conditions:
\begin{equation}
\label{eqn:projection_integrable}
\mathbb{E}\biggl[\int_0^T\int_0^{+\infty}\nu(dt,dz)\biggr] = \mathbb{E}\biggl[\int_0^T \lambda(t,Y_t)dt\biggr] <+\infty,
\end{equation}
and
\begin{equation}
\label{eqn:exp_z_finite}
\mathbb{E}\biggl[\int_0^T\int_0^{+\infty}{e^z\lambda(t,Y_t) dF(z, Y_t)\,dt}\biggr]<\infty,
\end{equation}
which implies the following:
\begin{equation*}
\mathbb{E}\biggl[\int_0^T\int_0^{+ \infty} z\, \lambda(t,Y_t) dF(z, Y_t) dt\biggr]<+\infty  ,
\qquad \mathbb{E}\biggl[ \int_0^T\int_0^{+ \infty} z^2\, \lambda(t,Y_t) dF(z, Y_t) dt \biggr]<+\infty.
\end{equation*}

According with the definition of dual predictable projection, for every nonnegative, $\{\mathcal{F}_t\}_{t\in[0,T]}$-predictable and $[0,+\infty)$-indexed process $\{H(t,z)\}_{t\in[0,T]}$ we have that%
\footnote{For details on marked point processes theory, see~\cite{bremaud:pointproc}.}
\begin{equation}
\label{pp}
\mathbb{E}\biggl[\int_0^T\int_0^{+ \infty}H(t,z)\,m(dt,dz)\biggr]=\mathbb{E}\biggl[\int_0^T\int_0^{+ \infty} H(t,z)\,\lambda(t,Y_t) dF(z, Y_t)\,dt\biggr]  .
\end{equation}
In particular, choosing $H(t,z) = H_t$ with $\{H_t\}_{t\in[0,T]}$ any nonnegative $\{\mathcal{F}_t\}_{t\in[0,T]}$-predictable process 
\[
\mathbb{E}\biggl[\int_0^T\int_0^{+ \infty}H(t) \,m(dt,dz)\biggr]=\mathbb{E}\biggl[\int_0^T  H(t) dN_t \biggr]  = \mathbb{E}\biggl[\int_0^TH_t \,\lambda(t,Y_t) dt\biggr] ,
\]
i.e. the claims arrival process $N_t = m((0,t] \times [0,+\infty)) = \sum_{n\ge1} \mathbbm{1}_{\{T_n\le t\}}$ is a point process with stochastic intensity $\{\lambda(t,Y_t)\}_{t\in[0,T]}$. 

Now we give the interpretation of $F(z, Y_t)$ as conditional distribution of the claim sizes%
\footnote{This result is an extension of Proposition 2.4 in~\cite{CG2006}.}.
 
\begin{Proposition}
\label{prop:F_distribution}
$\forall n=1, \dots$ and $\forall A\in \mathcal{B}([0,+\infty)) $

$$\mathbb{P}[Z_n \in A\mid \mathcal{F}_{t^-} ] = \int_A dF(z, Y_t) \quad dt \times d\mathbb{P}-a.s..$$

In particular, this implies that
\[
\mathbb{P}[Z_n \in A\mid \mathcal{F}_{T_n^-} ]= \mathbb{P}[Z_n \in A\mid \mathcal{F}^Y_{T_n} ]=\int_A dF(z, Y_{T_n}) \qquad  \text{a.s.},
\]
where $\mathcal{F}_{T_n^-}$ is the strict past of the $\sigma$-algebra generated by the stopping time $T_n$:
$$\mathcal{F}_{T_n^-} := \sigma\{ A \cap \{ t < \tau_n\},  A \in \mathcal{F}_t,  t \in [0,T]\}.$$
 \end{Proposition}
\begin{proof}
See Appendix~\ref{appendix:proofs}.
\end{proof}

This means that in our model both the claim arrival intensity and the claim size distribution are affected by the stochastic  factor $Y$. This is a reasonable assumption; for example, in automobile insurance $Y$ may describe weather, road conditions, traffic volume, and so on. For a detailed discussion of this topic see also \cite{BC:IME2018}.

\begin{Remark}
\label{remark:random_measure}
Let us observe that for any $\{\mathcal{F}_t\}_{t\in[0,T]}$-predictable and $[0,D]$-indexed process $\{H(t,z)\}_{t\in[0,T]}$ such that
\[
\mathbb{E}\biggl[\int_0^T\int_0^{+\infty} \abs{H(t,z)}\,\lambda(t,Y_t) dF(z, Y_t)\,dt\biggr]<\infty ,
\]
the process
\[
M_t = \int_0^t\int_0^{+\infty} H(s,z)\,\bigl(m(ds,dz)-\nu(ds,dz)\bigr) \qquad t\in[0,T]
\]
turns out to be an $\{\mathcal{F}_t\}_{t\in[0,T]}$-martingale. If in addition
\[
\mathbb{E}\biggl[\int_0^T\int_0^{+\infty} \abs{H(t,z)}^2\,\lambda(t,Y_t) dF(z, Y_t)\,dt\biggr]<\infty ,
\]
then $\{M_t\}_{t\in [0,T]}$ is a square integrable $\{\mathcal{F}_t\}_{t\in[0,T]}$-martingale and
\[
\mathbb{E}[M_t^2] =\mathbb{E}\biggl[\int_0^t\int_0^{+\infty} \abs{H(s,z)}^2\,\lambda(s,Y_s) dF(z, Y_s)\,ds\biggr] \qquad\forall t\in[0,T] .
\]
Moreover, the predictable covariation process of $\{M_t\}_{t\in[0,T]}$ is given by
\[
\langle M\rangle_t = \int_0^t\int_0^D \abs{H(s,z)}^2\,\lambda(s,Y_s) dF(z, Y_s)\,ds  ,
\]
that is $\{M_t^2 - \langle M\rangle_t\}_{t\in[0,T]}$ is an $\{\mathcal{F}_t\}_{t\in[0,T]}$-martingale%
\footnote{For these results and other related topics see e.g.~\cite{bass:2004}.}.
\end{Remark}

In this framework we define the cumulative claims up to time $t\in[0,T]$ as follows
\[
C_t= \sum_{n\ge1} Z_n\mathbbm{1}_{\{T_n\le t\}} = \int_0^t\int_0^{+\infty}zm(ds,dz)  ,
\]
and the reserve process of the insurance is described by
\[
R_t=R_0 + \int_0^t c_s ds - C_t  , 
\]
where $R_0>0$ is the initial wealth and  $\{c_t\}_{t\in[0,T]}$ is a non negative $\{\mathcal{F}_t\}_{t\in[0,T]}$-adapted process representing the gross insurance risk premium. In the sequel we assume $c_t = c(t,Y_t)$, for a suitable function $c(t,y)$ such that  $\mathbb{E}\bigl[\int_0^T c(t,Y_t) dt\bigr] <+\infty$. 

Now we allow the insurer to buy an excess-of-loss reinsurance contract. By means of this agreement, the insurer chooses a retention level $\alpha\in[0,+\infty)$ and for any future claim the reinsurer is responsible for all the amount which exceeds that threshold $\alpha$ (e.g. $\alpha=0$ means full reinsurance).
For any dynamic reinsurance strategy $\{\alpha_t\}_{t\in[0,T]}$, the insurer's surplus process is given by
\[
R^\alpha_t = R_0 + \int_0^t(c_s-q^\alpha_s )\,ds - \int_0^t\int_0^{+\infty}(z\land\alpha_s)\,m(ds,dz)  ,
\]
where $\{q^\alpha_t\}_{t\in[0,T]}$ is a non negative $\{\mathcal{F}_t\}_{t\in[0,T]}$-adapted process representing the reinsurance premium rate. In addition, we suppose that the following assumption holds true.

\begin{assumption}(Excess.of-loss reinsurance premium)
\label{def:xl_premium}
Let us assume that for any reinsurance strategy $\{\alpha_t\}_{t\in[0,T]}$ the corresponding reinsurance premium process $\{q^\alpha_t\}_{t\in[0,T]}$ admits the following representation:
\[
q^\alpha_t = q(t,Y_t,\alpha_t) \qquad \forall\omega\in\Omega,t\in[0,T]  ,
\]
where $q(t,y,\alpha):[0,T]\times\mathbb{R}\times[0,+\infty)\to[0,+\infty)$ is a continuous function in $\alpha$, with continuous partial derivatives $\frac{\partial q(t,y,\alpha)}{\partial \alpha},\frac{\partial^2 q(t,y,\alpha)}{\partial \alpha^2}$ in $\alpha\in[0,+\infty)$, such that
\begin{enumerate}
\item $\frac{\partial q(t,y,\alpha)}{\partial \alpha}\le0$ for all $(t,y,\alpha)\in[0,T]\times\mathbb{R}\times[0,+\infty)$, since the premium is increasing with respect to the protection level;
\item $q(t,y,0)>c(t,y)$ $\forall (t,y)\in[0,T]\times\mathbb{R}$, because the cedant is not allowed to gain a profit without risk.
\end{enumerate}
In the rest of the paper, $\frac{\partial q(t,y,0)}{\partial \alpha}$ should be intended as a right derivative.
\end{assumption}

Assumption~\ref{def:xl_premium} formalizes the minimal requirements for a process $\{q^\alpha_t\}_{t\in[0,T]}$ to be a reinsurance premium. In the next examples we briefly recall the most famous premium calculation principles, because they are widely used in optimal reinsurance problems solving. In Appendix \ref{appendix:premium} the reader can find a rigorous derivation of the following formulas~\eqref{eqn:EVP} and~\eqref{eqn:VP}.

\begin{Example}
\label{example:evp}
The most famous premium calculation principle is the \textit{expected value principle} (abbr. EVP)%
\footnote{See \cite{young:premium_princ}.}. The underlying conjecture is that the reinsurer evaluates her premium in order to cover the expected losses plus a load which depends on the expected losses. In our framework, under the EVP the reinsurance premium is given by the following expression:
\begin{equation}
\label{eqn:EVP}
q(t,y,\alpha)=(1+\theta)\lambda(t,y)\int_0^{+\infty} (z - z\land\alpha)\,dF(z,y) ,
\end{equation}
for some safety loading $\theta>0$.
\end{Example}

\begin{Example}
\label{example:vp}
Another important premium calculation principle is the \textit{variance premium principle} (abbr. VP). In this case, the reinsurer's loading is proportional to the variance of the losses. More formally, the reinsurance premium admits the following representation:
\begin{equation}
\label{eqn:VP}
q(t,y,\alpha)=\lambda(t,y)\int_0^{+\infty} (z - z\land\alpha)\,dF(z,y)+\theta\lambda(t,y)\int_0^{+\infty} (z - z\land\alpha)^2\,dF(z,y) ,
\end{equation}
for some safety loading $\theta>0$.
\end{Example}

From now on we assume the following condition:
\begin{equation}
\label{eqn:q_exp}
\mathbb{E}\biggl[e^{\eta \int_0^T e^{r (T- s)}q(s,Y_s,0)\,ds}\biggr]<+\infty
\end{equation}

Furthermore, the insurer can lend or borrow money at a fixed interest rate $r\in\mathbb{R}$. More precisely, every time the surplus is positive, the insurer lends it and earns interest income if $r>0$ (or pays interest expense if $r<0$); on the contrary, when the surplus becomes negative, the insurer borrows money and pays interest expense (or gains interest income if $r<0$).

Under these assumptions, the total wealth dynamic associated with a given strategy $\alpha$ is described by the following SDE:
\begin{equation}
\label{eqn:wealth_SDE}
dX^\alpha_t = dR^\alpha_t + rX^\alpha_t\,dt ,  \qquad X^\alpha_0=R_0 .
\end{equation}
It can be verified that the solution to~\eqref{eqn:wealth_SDE} is given by the following expression:
\begin{equation}
\label{eqn:wealth_sol}
X^\alpha_t = R_0 e^{r t} + \int_0^t e^{r (t- s)}\bigl[c(s,Y_s)-q(s,Y_s, \alpha_s)\bigr]\,ds
-\int_0^t\int_0^{+ \infty}  e^{ r (t-s)}(z\land\alpha_s)\,m(ds,dz) .
\end{equation}

Our aim is to find the optimal strategy $\alpha$ in order to maximize the expected exponential utility of the terminal wealth, that is
\[
\sup_{\alpha\in\mathcal{A}}{\mathbb{E}\biggl[1-e^{-\eta X^\alpha_T}\biggr]}=1-\inf_{\alpha\in\mathcal{A}}{\mathbb{E}\biggl[e^{-\eta X^\alpha_T}\biggr]} ,
\]
where $\eta>0$ is the risk-aversion parameter and $\mathcal{A}$ is the set of all admissible strategies as defined below.

\begin{Definition}
\label{def:admissible_strategies}
We denote by $\mathcal{A}$ the set of all admissible strategies, that is the class of all non negative $\{\mathcal{F}_t\}_{t\in[0,T]}$-predictable processes $\alpha_t$. With the notation $\mathcal{A}_t$ we refer to the same class, restricted to the strategies starting from $t\in[0,T]$.
\end{Definition}

\begin{Remark}
Observe that the condition~\eqref{eqn:q_exp} implies $\mathbb{E}\bigl[e^{-\eta X^\alpha_T}\bigr]<+\infty$ $\forall\alpha\in\mathcal{A}$. In fact we have that
\begin{align*}
\mathbb{E}\bigl[e^{-\eta X^\alpha_T}\bigr]&=\mathbb{E}\biggl[e^{-\eta R_0 e^{rT} -\eta \int_0^T e^{r (T- s)}\bigl[c(s,Y_s)-q(s,Y_s,\alpha_s)\bigr]\,ds-\eta\int_0^T\int_0^{+ \infty}  e^{ r (T-s)}(z\land\alpha_s)\,m(ds,dz)}\biggr]\\
&\le \mathbb{E}\biggl[e^{\eta \int_0^T e^{r (T- s)}q(s,Y_s,0)\,ds}\biggr].
\end{align*}
\end{Remark}

As usual in stochastic control problems, we focus on the corresponding dynamic problem:
\begin{equation}
\label{eqn:dynamic_problem}
\essinf_{\alpha\in\mathcal{A}_t}{\mathbb{E}\biggl[e^{-\eta X^\alpha_{t,x}(T)}\mid\mathcal{F}_t\biggr]} , 
\qquad t\in[0,T] ,
\end{equation}
where $X^\alpha_{t,x}(T)$ denotes the insurer's wealth process starting from $(t,x)\in[0,T]\times\mathbb{R}$ evaluated at time $T$.


\section{HJB formulation}
\label{section:HJB}

In order to solve the optimization problem~\eqref{eqn:dynamic_problem}, we introduce the \textit{value function} $v:[0,T]\times\mathbb{R}^2\to(0,+\infty)$ associated with it, that is
\begin{equation}
\label{eqn:value_fun}
v(t,x,y)\doteq\inf_{\alpha\in\mathcal{A}_t}{\mathbb{E}\biggl[e^{-\eta X^\alpha_{t,x}(T)}\mid Y_t=y \biggr]} .
\end{equation}
This function is expected to solve the Hamilton-Jacobi-Bellman (HJB) equation:
\begin{equation}
\label{eqn:HJBformulation_prop}
\left\{
\begin{aligned}
&\inf_{\alpha\in [0,+\infty)}{\mathcal{L}^\alpha v(t,x,y)} = 0
\qquad \forall(t,x,y)\in[0,T]\times\mathbb{R}^2\\
& v(T,x,y)= e^{-\eta x} \qquad \forall(x,y)\in\mathbb{R}^2 ,
\end{aligned}
\right.
\end{equation}
where $\mathcal{L}^\alpha$ denotes the Markov generator of the couple $(X^\alpha_t,Y_t)$ associated with a constant control $\alpha$.
In what follows, we denote by $\mathcal{C}^{1,2}_b$ the class of all bounded functions $f(t,x_1,\dots,x_n)$, with $n\ge1$, with bounded first order derivatives $\frac{\partial{f}}{\partial{t}},\frac{\partial{f}}{\partial{x_1}},\dots,\frac{\partial{f}}{\partial{x_n}}$ and bounded second order derivatives with respect to the spatial variables $\frac{\partial^2{f}}{\partial{x_1^2}},\dots,\frac{\partial{f}}{\partial{x_n^2}}$.

\begin{Lemma}
\label{lemmagenerator}
Let $f:[0,T]\times\mathbb{R}^2\to\mathbb{R}$ be a function in $\mathcal{C}^{1,2}_b$. The Markov generator of the stochastic process $(X^\alpha_t,Y_t)$ for all constant strategies $\alpha\in[0,+\infty)$ is given by the following expression:
\begin{multline}
\label{eqn:generator}
\mathcal{L}^\alpha f(t,x,y) = \frac{\partial{f}}{\partial{t}}(t,x,y)
+ \frac{\partial{f}}{\partial{x}}(t,x,y) \bigl[rx+c(t,y)-q(t,y,\alpha)\bigr]
+ b(t,y)\frac{\partial{f}}{\partial{y}}(t,x,y)\\
+ \frac{1}{2}\gamma(t,y)^2\frac{\partial^2{f}}{\partial{y^2}}(t,x,y) 
+ \int_0^{+\infty}{\biggl[f(t,x-z\land\alpha,y)-f(t,x,y)\biggr]\lambda(t,y)\,dF(z,y)} .
\end{multline}
\end{Lemma}
\begin{proof}
For any $f\in\mathcal{C}^{1,2}_b$, applying It\^o's formula to the stochastic process $f(t,X^\alpha_t,Y_t)$, we get the following expression:
\[
f(t,X^\alpha_t,Y_t)=f(0,X^\alpha_0,Y_0)+\int_0^t\mathcal{L}^\alpha f(s,X^\alpha_s,Y_s)\,ds + M_t ,
\]
where $\mathcal{L}^\alpha$ is defined in~\eqref{eqn:generator} and
\begin{multline*}
M_t=\int_0^t\gamma(s,Y_s)\frac{\partial{f}}{\partial{y}}(s,X^\alpha_s,Y_s)\,dW^{(Y)}_s\\
+\int_0^t\int_0^{+\infty}{\biggl(f(s,X^\alpha_s-z\land\alpha,Y_s)-f(s,X^\alpha_s,Y_s)\biggr)\bigl(m(ds,dz)-\nu(ds,dz)\bigr)} .
\end{multline*}
In order to complete the proof, we have to show that $\{M_t\}_{t\in[0,T]}$ is an $\{\mathcal{F}_t\}_{t\in[0,T]}$-martingale. For the first term, we observe that
\[
\mathbb{E}\biggl[\int_0^t\gamma(s,Y_s)^2\biggl(\frac{\partial{f}}{\partial{y}}(s,X^\alpha_s,Y_s)\biggr)^2\,ds\biggr]<\infty ,
\]
because the partial derivative is bounded and using the assumption~\eqref{eqn:solutionY}. For the second term, it is sufficient to use the boundedness of $f$ and the condition~\eqref{eqn:projection_integrable}.
\end{proof}

\begin{Remark}
Since the couple $(X^\alpha_t,Y_t)$ is a Markov process, any Markovian control is of the form $\alpha_t=\alpha(t,X^\alpha_t,Y_t)$, where $\alpha(t,x,y)$ denotes a suitable function. The generator $\mathcal{L}^\alpha f(t,x,y)$ associated to a general Markovian strategy can be easily obtained by replacing $\alpha$ with $\alpha_t$ in~\eqref{eqn:generator}.
\end{Remark}

In order to simplify our optimization problem, we present a preliminary result.
\begin{Remark}
\label{remark:survival_F}
Let $g:\mathbb{R}\mapsto[0,+\infty)$ be an integrable function such that $g(0)=0$. For any $\alpha\in[0,+\infty)$, the following equation holds true:
\[
\int_0^{+\infty}{g(z\land\alpha)\,dF(z,y)} = \int_0^{\alpha}{g'(z)\bar{F}(z,y)\,dz}
\qquad\forall y\in\mathbb{R} ,
\]
where $\bar{F}(z,y)\doteq 1-F(z,y)$. In fact, by integration by parts we get that
\begin{align}
\int_0^{+\infty}{g(z\land\alpha)\,dF(z,y)} &= \int_0^{\alpha}{g(z)\,dF(z,y)} + \int_{\alpha}^{+\infty}{g(\alpha)\,dF(z,y)} \notag\\
	&= g(\alpha)F(\alpha,y)- g(0)F(0,y) -\int_0^{\alpha}{g'(z)F(z,y)\,dz} + g(\alpha)[1-F(\alpha,y)] \notag\\
	&= -\int_0^{\alpha}{g'(z)F(z,y)\,dz} + \int_0^{\alpha}{g'(z)\,dz} \notag\\
	&= \int_0^{\alpha}{g'(z)(1-F(z,y))\,dz} .
\end{align}
\end{Remark}

Now let us consider the ansatz $v(t,x,y) = e^{-\eta xe^{r(T-t)}}\varphi(t,y)$, which is motivated by the following proposition.
 \begin{Proposition}
\label{prop:phi_to_v}
Let us suppose that there exists a function $\varphi:[0,T]\times\mathbb{R}\to(0,+\infty)$ solution to the following Cauchy problem:
\begin{multline}
\label{eqn:hjb_prop}
\frac{\partial{\varphi}}{\partial{t}}(t,y) + b(t,y)\frac{\partial{\varphi}}{\partial{y}}(t,y) + \frac{1}{2}\gamma(t,y)^2 \frac{\partial^2{\varphi}}{\partial{y^2}}(t,y)\\
+ \eta e^{ r (T-t)}\varphi(t,y)\bigl[-c(t,y)+\inf_{\alpha\in[0,+\infty)}{\Psi^{\alpha}(t,y)}\bigr]= 0 ,
\end{multline}
with final condition $\varphi(T,y) = 1$, $\forall y\in\mathbb{R}$,  where 
\begin{equation}
\label{eqn:Psi}
\Psi^{\alpha}(t,y) \doteq q(t,y,\alpha) + \lambda(t,y)\int_0^\alpha{e^{\eta ze^{r(T-t)}}\bar{F}(z,y)\,dz} ,
\qquad\alpha\in[0,+\infty) .
\end{equation}
Then the function
\begin{equation}
\label{eqn:v_from_phi}
v(t,x,y) = e^{-\eta xe^{r(T-t)}}\varphi(t,y)
\end{equation}
solves the HJB problem given in~\eqref{eqn:HJBformulation_prop}.
 \end{Proposition}
\begin{proof}
From the expression~\eqref{eqn:v_from_phi} we can easily verify that
\begin{align*}
e^{\eta xe^{r(T-t)}}\mathcal{L}^\alpha v(t,x,y)
&=\frac{\partial{\varphi}}{\partial{t}}(t,y)-\eta e^{ r (T-t)}\varphi(t,y)\bigl[c(t,y)-q(t,y,\alpha)\bigr]\\
&+ b(t,y)\frac{\partial{\varphi}}{\partial{y}}(t,y)+ \frac{1}{2}\gamma(t,y)^2 \frac{\partial^2{\varphi}}{\partial{y^2}}(t,y)\\
&+ \int_0^{+\infty}{\biggl[e^{\eta (z\land\alpha) e^{r(T-t)}}\varphi(t,y)-\varphi(t,y)\biggr]\lambda(t,y)\,dF(z,y)} .
\end{align*}
By Remark~\ref{remark:survival_F}, taking $g(z)=e^{\eta ze^{r(T-t)}}-1$, we can rewrite the last integral in this more convenient way:
\begin{align*}
&\varphi(t,y)\lambda(t,y)\int_0^{+\infty}{\biggl[e^{\eta (z\land\alpha)e^{r(T-t)}}-1\biggr]\,dF(z,y)}=\varphi(t,y)\lambda(t,y)\int_0^{\alpha}{\eta e^{r(T-t)}e^{\eta ze^{r(T-t)}}\,\bar{F}(z,y)\,dz} .
\end{align*}
Now we define $\Psi^{\alpha}(t,y)$ by means of the equation~\eqref{eqn:Psi}, obtaining the following equivalent expression:
\begin{align*}
e^{\eta xe^{r(T-t)}}\mathcal{L}^\alpha v(t,x,y)
&=\frac{\partial{\varphi}}{\partial{t}}(t,y)-\eta e^{ r (T-t)}\varphi(t,y)c(t,y)\\
&+ b(t,y)\frac{\partial{\varphi}}{\partial{y}}(t,y)+ \frac{1}{2}\gamma(t,y)^2 \frac{\partial^2{\varphi}}{\partial{y^2}}(t,y)+\eta e^{r (T-t)}\varphi(t,y)\Psi^{\alpha}(t,y) .
\end{align*}
Taking the infimum over $\alpha\in[0,+\infty)$, by~\eqref{eqn:hjb_prop} we find out the PDE in~\eqref{eqn:HJBformulation_prop}. The terminal condition in~\eqref{eqn:HJBformulation_prop} immediately follows by definition.
\end{proof}

The previous result suggests to focus on the minimization of the function~\eqref{eqn:Psi}, that is the aim of the next section.


\section{Optimal reinsurance strategy}
\label{section:reinsurance}

In this section we study the following minimization problem:
\begin{equation}
\label{eqn:optimal_reins_pb}
\inf_{\alpha\in[0,+\infty)}{\Psi^{\alpha}(t,y)},
\end{equation}
where $\Psi^{\alpha}(t,y):[0,T]\times\mathbb{R}\to(0,+\infty)$ is defined in~\eqref{eqn:Psi}.

In particular, we provide a complete characterization of the optimal reinsurance strategy. In the sequel we assume $0\le F(z,y)<1$ $\forall(z,y)\in[0,+\infty)\times\mathbb{R}$.

 \begin{Proposition}
\label{prop:optimal_XL}
Let us suppose that $\Psi^{\alpha}(t,y)$ is strictly convex in $\alpha\in[0,+\infty)$ and let us define the set $A_0\subseteq[0,T]\times\mathbb{R}$ as follows:
\begin{equation}
\label{eqn:set_A0}
A_0\doteq\Set{(t,y)\in [0,T]\times\mathbb{R}\mid-\frac{\partial{q(t,y,0)}}{\partial{\alpha}}\le\lambda(t,y)}.
\end{equation}
If the equation
\begin{equation}
\label{eqn:first_order_cond}
-\frac{\partial{q(t,y,\alpha)}}{\partial{\alpha}}=\lambda(t,y)e^{\eta \alpha e^{r(T-t)}}\bar{F}(\alpha,y)
\end{equation}
admits at least one solution in $(0,+\infty)$ for any $(t,y)\in [0,T]\times\mathbb{R}\setminus A_0$, denoted by $\hat{\alpha}(t,y)$, then the minimization problem~\eqref{eqn:optimal_reins_pb} admits a unique solution $\alpha^*(t,y)\in[0,+\infty)$ given by
\begin{equation}
\label{eqn:optimal_solution}
\alpha^*(t,y)=
\begin{cases}
	0 & \text{$(t,y)\in A_0$}
	\\
	\hat{\alpha}(t,y) & \text{$(t,y)\in [0,T]\times\mathbb{R}\setminus A_0$.}
\end{cases}
\end{equation}
 \end{Proposition}
\begin{proof}
The function $\Psi^{\alpha}(t,y)$ is continuous in $\alpha\in[0,+\infty)$ by definition (see Assumption~\ref{def:xl_premium}) and for any $(t,y)\in[0,T]\times\mathbb{R}$ its derivative is given by the following expression:
\begin{equation}
\label{derivata}
\frac{\partial{\Psi^{\alpha}(t,y)}}{\partial{\alpha}}=\frac{\partial{q(t,y,\alpha)}}{\partial{\alpha}}
+\lambda(t,y)e^{\eta \alpha e^{r(T-t)}}\bar{F}(\alpha,y).
\end{equation}
Since $\Psi^{\alpha}(t,y)$ is convex in $\alpha\in[0,+\infty)$ by hypothesis, if $(t,y)\in A_0$ then $\frac{\partial{\Psi^{0}(t,y)}}{\partial{\alpha}}\ge0$, then $\alpha^*(t,y)=0$, because the derivative $\frac{\partial{\Psi^{\alpha}(t,y)}}{\partial{\alpha}}$ is increasing in $\alpha$ and there is no stationary point in $(0,+\infty)$. Else, if $(t,y)\in [0,T]\times\mathbb{R}\setminus A_0$ then $\frac{\partial{\Psi^{0}(t,y)}}{\partial{\alpha}}<0$, and $\alpha^*(t,y)$ coincides with the unique stationary point of $\Psi^{\alpha}(t,y)$, which is $\hat{\alpha}(t,y)\in(0,+\infty)$. Let us notice that it exists by hypothesis and it is unique because $\Psi^{\alpha}(t,y)$ is strictly convex.
\end{proof}

By the previous proposition, we observe that $\lambda(t,y)$ is an important threshold for the insurer: as long as the marginal cost of the full reinsurance falls in the interval $(0,\lambda(t,y)]$, the optimal choice is full reinsurance.

Unfortunately, it is not always easy to check whether $\Psi^{\alpha}(t,y)$ is strictly convex in $\alpha\in[0,+\infty)$ or not. In the next result such an hypothesis is relaxed, while the uniqueness of the solution to \eqref{eqn:first_order_cond} is required.

 \begin{Proposition}
\label{eqn:optimal_solution2}
Suppose that the equation~\eqref{eqn:first_order_cond} admits a unique solution $\hat{\alpha}(t,y)\in(0,+\infty)$ for any $(t,y)\in [0,T]\times\mathbb{R}\setminus A_0$. Moreover, let us assume that
\begin{equation}
\label{eqn:psi_convex_d2}
\frac{\partial^2{q(t,y,\hat{\alpha}(t,y))}}{\partial{\alpha^2}}
>-\lambda(t,y)e^{\eta\hat{\alpha}(t,y) e^{r(T-t)}}\frac{\partial{\bar{F}(\hat{\alpha}(t,y),y)}}{\partial{z}}
\quad \forall(t,y)\in [0,T]\times\mathbb{R}\setminus A_0.
\end{equation}
Then the minimization problem~\eqref{eqn:optimal_reins_pb} admits a unique solution $\alpha^*(t,y)\in[0,+\infty)$ given by \eqref{eqn:optimal_solution}.
 \end{Proposition}
\begin{proof}
Recalling the proof of Proposition~\ref{prop:optimal_XL}, if $(t,y)\in A_0$ then $\frac{\partial{\Psi^{0}(t,y)}}{\partial{\alpha}}\ge0$ and $\alpha^*(t,y)=0$. For any $(t,y)\in [0,T]\times\mathbb{R}\setminus A_0$, by hypothesis there exists a unique stationary point $\hat{\alpha}(t,y)\in(0,+\infty)$. By simple calculations, using \eqref{eqn:psi_convex_d2} we notice that
\[
\frac{\partial^2{\Psi^{\hat{\alpha}}(t,y)}}{\partial{\alpha^2}}>0,
\]
hence $\hat{\alpha}(t,y)$ is the unique minimizer and this completes the proof.
\end{proof}

The next result deals with the existence of a solution to~\eqref{eqn:first_order_cond}. In particular, it is sufficient to require that the claim size distribution is heavy-tailed, which is a relevant case in non-life insurance (see \cite[Chapter 2]{rolski:insurancefin}), plus a technical condition for the reinsurance premium.

 \begin{Proposition}
\label{prop:existence_suff}
Let us assume that the reinsurance premium $q(t,y,\alpha)$ is such that%
\footnote{E.g. if $q$ is convex in $\alpha$.}
\[
\lim_{\alpha\to+\infty}{\frac{\partial{q(t,y,\alpha)}}{\partial{\alpha}}}=l \in\mathbb{R}
\]
and the claim size distribution is heavy-tailed in this sense:
\[
\int_0^{+\infty}e^{kz}\,dF(z,y)=+\infty \qquad \forall k>0,y\in\mathbb{R}.
\]
Then, for any $(t,y)\in [0,T]\times\mathbb{R}\setminus A_0$, the equation~\eqref{eqn:first_order_cond} admits at least one solution in $(0,+\infty)$.
 \end{Proposition}
\begin{proof}
The following property of heavy-tailed distributions is a well known implication of our assumption:
\[
\lim_{z\to+\infty}{e^{kz}\bar{F}(z,y)}=+\infty \qquad \forall k>0,y\in\mathbb{R}.
\]
Hence, by equation \eqref{derivata},  for any $(t,y)\in[0,T]\times\mathbb{R}\setminus A_0$
\begin{align*}
\lim_{\alpha\to+\infty}{\frac{\partial{\Psi^{\alpha}(t,y)}}{\partial{\alpha}}}
&=\lim_{\alpha\to+\infty}{\biggl[\frac{\partial{q(t,y,\alpha)}}{\partial{\alpha}}
+\lambda(t,y)e^{\eta \alpha e^{r(T-t)}}\bar{F}(\alpha,y)\biggr]}=+\infty.
\end{align*}
On the other hand, we know that
\[
\frac{\partial{\Psi^{0}(t,y)}}{\partial{\alpha}}<0 \qquad \forall(t,y)\in[0,T]\times\mathbb{R}\setminus A_0.
\]
As a consequence, $\frac{\partial{\Psi^{\alpha}(t,y)}}{\partial{\alpha}}$ being continuous in $\alpha\in[0,+\infty)$, there exists $\hat{\alpha}(t,y)\in(0,+\infty)$ such that $\frac{\partial{\Psi^{\hat{\alpha}}(t,y)}}{\partial{\alpha}}=0$.
\end{proof}

Now we turn the attention to the other crucial hypothesis of Proposition~\ref{prop:optimal_XL}, which is the convexity of $\Psi^{\alpha}(t,y)$. The reader can easily observe that the reinsurance premium convexity plays a central role.

 \begin{Proposition}
Suppose that the reinsurance premium $q(t,y,\alpha)$ is convex in $\alpha\in[0,+\infty)$ and $F(z,y)=(1-e^{-\zeta(y) z})\mathbbm{1}_{\{z>0\}}$ for some function $\zeta(y)$ such that $0<\zeta(y)<\eta \min{\{e^{rT},1\}}$ $\forall y\in\mathbb{R}$.
Then the function $\Psi^{\alpha}(t,y)$ defined in~\eqref{eqn:Psi} is strictly convex in $\alpha\in[0,+\infty)$.
 \end{Proposition}
\begin{proof}
Recalling the expression~\eqref{eqn:Psi}, it is sufficient to prove the convexity of the following term:
\[
\int_0^\alpha{e^{\eta ze^{r(T-t)}}\bar{F}(z,y)\,dz}.
\]
For this purpose, let us evaluate its second order derivative:
\[
e^{\eta\alpha e^{r(T-t)}}\biggl(\eta e^{r(T-t)}\bar{F}(\alpha,y)+\frac{\partial{\bar{F}(\alpha,y)}}{\partial{z}}\biggr).
\]
Now the term in brackets is
\[
\eta e^{r(T-t)}e^{-\zeta(y) \alpha}-\zeta(y)e^{-\zeta(y) \alpha}>0 \qquad \forall t\in[0,T].
\]
The proof is complete.
\end{proof}

By Proposition~\ref{prop:F_distribution}, the hypothesis on the claim sizes distribution above may be read as assuming that the claims are exponentially distributed conditionally to $Y$.


\subsection{Expected value principle}

Now we investigate the special case of the expected value principle introduced in Example~\ref{example:evp}.

 \begin{Proposition}
\label{prop:evp}
Under the EVP (see equation~\eqref{eqn:EVP}), the optimal reinsurance strategy $\alpha^*(t)\in[0,+\infty)$ is given by
\begin{equation}
\label{eqn:optimal_reins_EVP}
\alpha^*(t)=e^{-r(T-t)}\frac{\log{(1+\theta)}}{\eta}, \qquad t\in[0,T].
\end{equation}
 \end{Proposition}
\begin{proof}
Using Remark~\ref{remark:survival_F}, we can rewrite the equation~\eqref{eqn:EVP} as follows:
\[
q(t,y,\alpha)=(1+\theta)\lambda(t,y)\biggl[\int_0^{+\infty}z\,dF(z,y)-\int_0^{\alpha}\bar{F}(z,y)\,dz\biggr].
\]
As a consequence, we have that
\[
\frac{\partial{q(t,y,\alpha)}}{\partial{\alpha}}=-(1+\theta)\lambda(t,y)\bar{F}(\alpha,y)
\qquad\forall \alpha\in[0,+\infty).
\]
For $\alpha=0$, we have that
\[
\frac{\partial{\Psi^{0}(t,y)}}{\partial{\alpha}}=\frac{\partial{q(t,y,0)}}{\partial{\alpha}}+\lambda(t,y)<0
\qquad\forall(t,y)\in[0,T]\times\mathbb{R},
\]
hence $A_0=\emptyset$ and by Proposition~\ref{prop:optimal_XL} the minimizer belongs to $(0,+\infty)$. Now we look for the stationary points, i.e. the solutions to the equation~\eqref{eqn:first_order_cond}, that in this case reads as follows:
\begin{equation}
\label{eqn:equation_EVP}
(1+\theta)\lambda(t,y)\bar{F}(\alpha,y)=\lambda(t,y)e^{\eta \alpha e^{r(T-t)}}\bar{F}(\alpha,y).
\end{equation}
Solving this equation, we obtain the unique solution given by~\eqref{eqn:optimal_reins_EVP}. In order to prove that it coincides with the unique minimizer to~\eqref{eqn:optimal_reins_pb}, it is sufficient to show that
\[
\frac{\partial^2{\Psi^{\alpha^*(t)}(t,y)}}{\partial{\alpha^2}}>0.
\]
For this purpose, observe that
\begin{align*}
\frac{\partial^2{\Psi^{\alpha^*(t)}(t,y)}}{\partial{\alpha^2}}
&=\frac{\partial^2{q(t,y,\alpha^*(t))}}{\partial{\alpha^2}}
+\lambda(t,y)e^{\eta\alpha^*(t) e^{r(T-t)}}\biggl(\eta e^{r(T-t)}\bar{F}(\alpha^*(t),y)+\frac{\partial{\bar{F}(\alpha^*(t),y)}}{\partial{z}}\biggr)\\
&>-(1+\theta)\lambda(t,y)\frac{\partial{\bar{F}(\alpha^*(t),y)}}{\partial{z}}
+\lambda(t,y)e^{\eta\alpha^*(t) e^{r(T-t)}}\frac{\partial{\bar{F}(\alpha^*(t),y)}}{\partial{z}}\\
&=0.
\end{align*}
The proof is complete.
\end{proof}

\begin{Remark}
Formula~\eqref{eqn:optimal_reins_EVP} was found by \cite{zhaoetal:2013} (see equation 3.31, page 508). We point out that it is a completely deterministic strategy. This fact is crucially related to the use of the EVP rather than the underlying model; in fact, in \cite{zhaoetal:2013} the authors considered the Cram\'er-Lundberg model under the EVP%
\footnote{It is not surprising, in fact in~\cite{BC:IME2018} and references therein also the optimal proportional reinsurance under EVP turns out to be deterministic}.
\end{Remark}

From the economic point of view, by equation~\eqref{eqn:optimal_reins_EVP} it is easy to show that the optimal retention level is decreasing with respect to the interest rate and the risk-aversion; on the contrary, it is increasing with respect to the reinsurer's safety loading. In addition, the sensitivity with respect to the time-to-maturity depends on the sign of $r$.

Another relevant aspect of~\eqref{eqn:optimal_reins_EVP} is that it is independent of the claim size distribution. To the authors this result seems quite unrealistic. In fact, any subscriber of an excess-of-loss contract is strongly worried about possibly extreme events, hence the claims distribution is expected to play an important role.


\subsection{Variance premium principle}

This subsection is devoted to derive an optimal strategy under the variance premium principle (see Example~\ref{example:vp}).

 \begin{Proposition}
\label{prop:vp}
Let us suppose that $\Psi^{\alpha}(t,y)$ is strictly convex in $\alpha\in[0,+\infty)$ and
\begin{equation}
\label{eqn:vp_distribution}
\lim_{z\to+\infty}{e^{\eta\min{\{e^{rT},1\}}z}}\bar{F}(z,y)=l,
\end{equation}
for some $l>0$ (eventually $l=+\infty$).

Under the VP (see equation~\eqref{eqn:VP}) the optimal reinsurance strategy $\alpha^*(t,y)$ is the unique solution to the following equation:
\begin{equation}
\label{eqn:vp_equation}
\bigl(e^{\eta \alpha e^{r(T-t)}}+2\theta\alpha-1\bigr)\bar{F}(\alpha,y)=2\theta\int_\alpha^{+\infty}z\,dF(z,y).
\end{equation}
 \end{Proposition}
\begin{proof}
The proof is based on Proposition~\eqref{prop:optimal_XL}. By equation~\eqref{eqn:VP} we get its derivative:
\[
\frac{\partial{q(t,y,\alpha)}}{\partial{\alpha}}=\lambda(t,y)\bar{F}(\alpha,y)(2\theta\alpha-1)-2\theta\lambda(t,y)\int_\alpha^{+\infty}z\,dF(z,y).
\]
It is clear that the set $A_0$ defined in~\eqref{eqn:set_A0} is empty, because for any $(t,y)\in[0,T]\times\mathbb{R}$
\[
-\frac{\partial{q(t,y,0)}}{\partial{\alpha}}=\lambda(t,y)\bar{F}(0,y)+2\theta\lambda(t,y)\int_0^{+\infty}z\,dF(z,y)>\lambda(t,y).
\]
Hence the minimizer should coincide with the unique stationary point of $\Psi^{\alpha}(t,y)$, i.e. the solution to~\eqref{eqn:vp_equation}. In order to prove it, we need to ensure the existence of a solution to~\eqref{eqn:vp_equation}. For this purpose, we notice that on the one hand
\[
\frac{\partial{\Psi^{0}(t,y)}}{\partial{\alpha}}=-2\theta\lambda(t,y)\int_0^{+\infty}z\,dF(z,y)<0.
\]
On the other hand, for $\alpha\to+\infty$, by~\eqref{eqn:vp_distribution} we get
\begin{align*}
\lim_{\alpha\to+\infty}{\frac{\partial{\Psi^{\alpha}(t,y)}}{\partial{\alpha}}}
&=\lambda(t,y)\lim_{\alpha\to+\infty}{\biggl[\bigl(e^{\eta \alpha e^{r(T-t)}}+2\theta\alpha-1\bigr)\bar{F}(\alpha,y)-2\theta\int_\alpha^{+\infty}z\,dF(z,y)\biggr]}>0.
\end{align*}
As a consequence, by the continuity of $\Psi^{\alpha}(t,y)$ there exists a point $\alpha^*\in(0,+\infty)$ such that $\frac{\partial{\Psi^{\alpha^*}(t,y)}}{\partial{\alpha}}=0$. Such a solution is unique because $\Psi^{\alpha}(t,y)$ is strictly convex by hypothesis.
\end{proof}

Conversely to Proposition~\ref{prop:evp},  the optimal retention level given in Proposition~\ref{prop:vp} is still dependent on the stochastic factor $Y$. Such a dependence is spread through the claim size distribution.

\begin{Remark}
We observe that any heavy-tailed distribution (see the proof of Proposition~\ref{prop:existence_suff}) satisfies the condition~\eqref{eqn:vp_distribution} with $l=+\infty$.
\end{Remark}

Now we specialize the variance premium principle to conditionally exponentially distributed claims.

 \begin{Proposition}
\label{prop:vp_exp}
Under the VP, suppose that $F(z,y)=(1-e^{-\zeta(y) z})\mathbbm{1}_{\{z>0\}}$ for some function $\zeta(y)$ such that $\zeta(y)>0$ $\forall y\in\mathbb{R}$. The optimal reinsurance strategy is given by
\begin{equation}
\label{eqn:optimal_reins_VP}
\alpha^*(t,y)=e^{-r(T-t)}\frac{\log{(1+\frac{2\theta}{\zeta(y)})}}{\eta}, \qquad (t,y)\in[0,T]\times\mathbb{R}.
\end{equation}
 \end{Proposition}
\begin{proof}
By the proof of Proposition~\ref{prop:vp}, we know that under VP $A_0=\emptyset$. Now, under our hypotheses, by equation \eqref{derivata} we readily get
\begin{align*}
\frac{\partial{\Psi^{\alpha}(t,y)}}{\partial{\alpha}}&=\lambda(t,y)\biggl[\bigl(e^{\eta \alpha e^{r(T-t)}}+2\theta\alpha-1\bigr)\bar{F}(\alpha,y)-2\theta\int_\alpha^{+\infty}z\,dF(z,y)\biggr]\\
&=\lambda(t,y)\biggl[\bigl(e^{\eta \alpha e^{r(T-t)}}+2\theta\alpha-1\bigr)e^{-\zeta(y) \alpha}-2\theta e^{-\zeta(y) \alpha}\bigl(\alpha+\frac{1}{\zeta(y)}\bigr)\biggr]\\
&=\lambda(t,y)e^{-\zeta(y) \alpha}\biggl[e^{\eta \alpha e^{r(T-t)}}-1-\frac{2\theta}{\zeta(y)}\biggr].
\end{align*}
The equation $\frac{\partial{\Psi^{\alpha}(t,y)}}{\partial{\alpha}}=0$ admits a unique solution, given by equation~\eqref{eqn:optimal_reins_VP}. At this point $\alpha^*(t,y)$, the function $\Psi^{\alpha}(t,y)$ is strictly convex, because
\begin{align*}
\frac{\partial{\Psi^{\alpha^*}(t,y)}}{\partial{\alpha}}&=-\zeta(y)\frac{\partial{\Psi^{\alpha^*}(t,y)}}{\partial{\alpha}}+\lambda(t,y)e^{-\zeta(y) \alpha^*}\eta e^{r(T-t)}e^{\eta \alpha^* e^{r(T-t)}}\\
&=\lambda(t,y)e^{-\zeta(y) \alpha^*}\eta e^{r(T-t)}e^{\eta \alpha^* e^{r(T-t)}}>0.
\end{align*}
It follows that $\alpha^*(t,y)$ is the unique minimizer by Proposition~\ref{eqn:optimal_solution}.
\end{proof}
Contrary to the equation~\eqref{eqn:optimal_reins_EVP}, the explicit formula~\eqref{eqn:optimal_reins_VP} keeps the dependence on the stochastic factor $Y$. In addition, the following result holds true.

\begin{Remark}
Suppose that $F(z,y)=(1-e^{-\zeta(y) z})\mathbbm{1}_{\{z>0\}}$ for some function $\zeta(y)$ such that $\zeta(y)>0$ $\forall y\in\mathbb{R}$. We consider two different reinsurance safety loadings $\theta_{\text{EVP}},\theta_{\text{VP}}>0$, referring to the EVP and VP, respectively. Moreover, let us denote by $\alpha_{\text{EVP}}^*(t)$ and $\alpha_{\text{VP}}^*(t,y)$ the optimal retention level under the EVP and VP, given in equations~\eqref{eqn:optimal_reins_EVP} and~\eqref{eqn:optimal_reins_VP}, respectively. It is easy to show that $\forall t\in[0,T]$
\begin{equation*}
\alpha_{\text{VP}}^*(t,y)
\begin{cases}
	>\alpha_{\text{EVP}}^*(t) & \forall y:\zeta(y)<\frac{2\theta_{\text{VP}}}{\theta_{\text{EVP}}}
	\\
	\le\alpha_{\text{EVP}}^*(t) & \text{otherwise.}
\end{cases}
\end{equation*}
From the practical point of view, as long as the stochastic factor fluctuations result in a rate parameter $\zeta(y)$ higher than the threshold $\frac{2\theta_{\text{VP}}}{\theta_{\text{EVP}}}$, the optimal retention level evaluated through the expected value principle turns out to be larger than the variance principle.
\end{Remark}


\section{Verification Theorem}
\label{section:verification}

\begin{Theorem}[Verification Theorem]
\label{theorem:verification}
Let us suppose that $\varphi:[0,T]\times\mathbb{R}\to(0,+\infty)$ is a bounded classical solution $\varphi\in\mathcal{C}^{1,2}((0,T)\times\mathbb{R})\cap\mathcal{C}([0,T]\times\mathbb{R})$ to the Cauchy problem~\eqref{eqn:hjb_prop}, such that
\begin{equation}
\label{eqn:phi_sublinear}
\abs*{\frac{\partial{\varphi}}{\partial{y}}(t,y)}\le C(1+\abs{y}^\beta) \qquad\forall(t,y)\in[0,T]\times\mathbb{R},
\end{equation}
for some constants $\beta,C>0$. Then the function $v(t,x,y) = e^{-\eta xe^{r(T-t)}}\varphi(t,y)$ (see equation~\eqref{eqn:v_from_phi}) is the value function in equation~\eqref{eqn:value_fun}. As a byproduct, the strategy $\alpha^*_t\doteq\alpha^*(t,Y_t)$ described in Proposition~\ref{prop:optimal_XL} is an optimal control.
\end{Theorem}
\begin{proof}
By Proposition~\ref{prop:phi_to_v}, the function $v(t,x,y)$ defined in equation~\eqref{eqn:v_from_phi} solves the HJB problem~\eqref{eqn:HJBformulation_prop}. Hence for any $(t,x,y)\in[0,T]\times\mathbb{R}^2$
\[
\mathcal{L}^\alpha v(s,X^\alpha_{t,x}(s),Y_{t,y}(s)) \ge 0 \qquad \forall s\in[t,T],\alpha\in\mathcal{A}_t,
\]
where $\{X^\alpha_{t,x}(s)\}_{s\in [t,T]}$ and $\{Y_{t,y}(s)\}_{s\in [t,T]}$ denote the solutions to~\eqref{eqn:wealth_SDE} and~\eqref{eqn:stochasticfactor} at time $s\in[t,T]$, starting from $(t,x)\in[0,T]\times\mathbb{R}$ and $(t,y)\in[0,T]\times\mathbb{R}$, respectively.\\
From It\^o's formula we get
\begin{equation}
\label{ito_verification}
v(T,X^\alpha_{t,x}(T),Y_{t,y}(T))=v(t,x,y)+\int_t^T\mathcal{L}^\alpha v(s,X^\alpha_{t,x}(s),Y_s)\,ds + M_T,
\end{equation}
with $\{M_r\}_{r\in [t,T]}$ defined by
\begin{multline}
\label{eqn:M_mg}
M_r=\int_t^r\gamma(s,Y_s)\frac{\partial{v}}{\partial{y}}(s,X^\alpha_{t,x}(s),Y_s)\,dW^{(Y)}_s\\
+\int_t^r\int_0^{+\infty}{\biggl(v(s,X^\alpha_{t,x}(s)-z\land\alpha,Y_s)-v(s,X^\alpha_{t,x}(s),Y_s)\biggr)\bigl(m(ds,dz)-\nu(ds,dz)\bigr)}.
\end{multline}
In order to show that $\{M_r\}_{r\in [t,T]}$ is an $\{\mathcal{F}_t\}_{t\in [0,T]}$-local-martingale, we use a localization argument, taking
\[
\tau_n\doteq\inf\{s\in[t,T]\mid X^\alpha_{t,x}(s)<-n \lor \abs{Y_s}>n\}, \qquad n\in\mathbb{N}.
\]
The reader can easily check that $\{\tau_n\}_{n\in\mathbb{N}}$ is a non decreasing sequence of stopping time such that $\lim_{n\to+\infty}{\tau_n}=+\infty$. For the diffusion term of $M_r$, using the assumptions~\eqref{eqn:phi_sublinear} and~\eqref{eqn:solutionY}, we notice that
\begin{align*}
&\mathbb{E}\biggl[\int_t^{T\land\tau_n}\gamma(s,Y_s)^2\biggl(\frac{\partial{v}}{\partial{y}}(s,X^\alpha_{t,x}(s),Y_s)\biggr)^2\,ds\biggr]\\
&=\mathbb{E}\biggl[\int_t^{T\land\tau_n}\gamma(s,Y_s)^2e^{-2\eta X^\alpha_{t,x}(s)e^{r(T-t)}}\biggl(\frac{\partial{\varphi}}{\partial{y}}(s,Y_s)\biggr)^2\,ds\biggr]\\
&\le C_n\,\mathbb{E}\biggl[\int_t^{T\land\tau_n}\gamma(s,Y_s)^2\,ds\biggr]<\infty \qquad \forall n\in\mathbb{N},
\end{align*}
where $C_n>0$ is a constant depending on $n$. For the jump term, by the condition \eqref{eqn:exp_z_finite} and Remark \ref{remark:random_measure}, we get
\begin{align*}
&\mathbb{E}\biggl[\int_t^{T\land\tau_n}\int_0^{+\infty}{\abs{v(s,X^\alpha_{t,x}(s)-z\land\alpha,Y_s)-v(s,X^\alpha_{t,x}(s),Y_s)} \nu(ds,dz)}\biggr]\\
&=\mathbb{E}\biggl[\int_t^{T\land\tau_n}\int_0^{+\infty}{\abs*{e^{-\eta X^\alpha_{t,x}(s)}(e^{z\land \alpha}-1)\varphi(s,Y_s)}\nu(ds,dz)}\biggr]\\
&\le \tilde{C}_n\,\mathbb{E}\biggl[\int_t^{T\land\tau_n}\int_0^{+\infty}{e^z\nu(ds,dz) }\biggr]<\infty,
\end{align*}
with $\tilde{C}_n$ denoting a positive constant dependent on $n$. Thus $\{M_r\}_{r\in [t,T]}$ turns out to be an $\{\mathcal{F}_t\}_{t\in [0,T]}$-local-martingale and $\{\tau_n\}_{n\in\mathbb{N}}$ is a localizing sequence for it. Now, taking the expectation of~\eqref{ito_verification} with $T\land\tau_n$ in place of $T$, we obtain that
\[
\mathbb{E}[v(T\land\tau_n,X^\alpha_{t,x}(T\land\tau_n),Y_{t,y}(T\land\tau_n))\mid\mathcal{F}_t]\ge v(t,x,y) \quad\forall(t,x,y)\in[0,\land\tau_n]\times\mathbb{R}^2,\alpha\in\mathcal{A}_t,n\in\mathbb{N}.
\]
Let us notice that
\[
\mathbb{E}[v(T\land\tau_n,X^\alpha_{t,x}(T\land\tau_n),Y_{t,y}(T\land\tau_n))^2]\le \tilde{C}e^{-2\eta ne^{r(T-t)}}\le \tilde{C},
\]
where $\tilde{C}>0$ is a constant. As a consequence, $\{v(T\land\tau_n,X^\alpha_{t,x}(T\land\tau_n),Y_{t,y}(T\land\tau_n))\}_{n\in\mathbb{N}}$ is a sequence of uniformly integrable random variables. By classical results in probability theory, it converges almost surely. Using the monotonicity and the boundedness of $\{\tau_n\}_{n\in\mathbb{N}}$, together with the non explosion of $\{X^\alpha_{t,x}(s)\}_{s\in [t,T]}$ and $\{Y_{t,y}(s)\}_{s\in [t,T]}$ (see~\eqref{eqn:wealth_sol} and~\eqref{eqn:solutionY2}), taking the limit for $n\to+\infty$ we conclude that
\begin{align*}
\mathbb{E}[v(T,X^\alpha_{t,x}(T),Y_{t,y}(T))\mid\mathcal{F}_t]
&=\lim_{n\to+\infty}{\mathbb{E}[v(T\land\tau_n,X^\alpha_{t,x}(T\land\tau_n),Y_{t,y}(T\land\tau_n))\mid\mathcal{F}_t]}\\
&\ge v(t,x,y) \qquad \forall t\in[0,T],\alpha\in\mathcal{A}_t.
\end{align*}
As a byproduct, since $\alpha^*(t,y)$ given in Proposition~\ref{prop:optimal_XL} realizes the infimum in~\eqref{eqn:optimal_reins_pb}, we have that $\mathcal{L}^{\alpha^*}v(t,x,y)=0$ and, replicating the calculations above, we obtain the equality
\[
\inf_{\alpha\in\mathcal{A}_t}{\mathbb{E}\biggl[e^{-\eta X^\alpha_{t,x}(T)}\mid Y_t=y \biggr]}=v(t,x,y),
\]
i.e. $\alpha^*_t\doteq\alpha^*(t,Y_t)$ is an optimal control.
\end{proof}

By Theorem~\ref{theorem:verification}, the value function~\eqref{eqn:value_fun} can be characterized as a transformation of the solution to the partial differential equation (PDE)~\eqref{eqn:hjb_prop}. Nevertheless, an explicit expression is not available, except for very special cases. The following result provides a probabilistic representation by means of the Feynman-Kac theorem.

 \begin{Proposition}
\label{prop:fk}
Suppose that $\varphi:[0,T]\times\mathbb{R}\to(0,+\infty)$ is a bounded classical solution $\varphi\in\mathcal{C}^{1,2}((0,T)\times\mathbb{R})\cap\mathcal{C}([0,T]\times\mathbb{R})$ to the Cauchy problem~\eqref{eqn:hjb_prop}, such that the condition~\eqref{eqn:phi_sublinear} is fulfilled. Then the value function~\eqref{eqn:value_fun} admits the following representation:
\begin{equation}
\label{eqn:feynmankac}
v(t,x,y) = e^{-\eta xe^{r(T-t)}}\,\mathbb{E}\biggl[e^{\int_t^T{\eta e^{R(T-s)}\bigl(\inf_{\alpha\in[0,+\infty)}{\Psi^{\alpha}(s,Y_s)-c(s,Y_s)}\bigr) \,ds}}\mid Y_t=y \biggr],
\end{equation}
where $\Psi^{\alpha}(t,y)$ is the function defined in~\eqref{eqn:Psi}.
 \end{Proposition}
\begin{proof}
The thesis immediately follows by Theorem~\ref{theorem:verification} and the Feynman-Kac representation of $\varphi(t,y)$.
\end{proof}

\begin{Remark}
We refer to~\cite{heath:pde} for existence and uniqueness of a solution to the PDE \eqref{eqn:hjb_prop}.
\end{Remark}


\section{Numerical results}
\label{section:sim}

In this section we show some numerical results, mostly based on Propositions \ref{prop:evp} and \ref{prop:vp_exp}. We assumed the following dynamic for the stochastic factor $Y$ for performing simulations:
\[
dY_t=0.3\,dt + 0.3\,dW^{(Y)}_t , \qquad Y_0=1.
\]
The $\{\mathcal{F}_t\}_{t\in[0,T]}$-dual predictable projection $\nu(dt,dz)$ (see equation \eqref{eqn:dual_projection}) is determined by these functions:
\begin{gather*}
\lambda(t,y)=\lambda_0e^{\frac{1}{2}y}, \qquad \lambda_0=0.1,\\
F(z,y)=(1-e^{-\zeta(y) z})\mathbbm{1}_{\{z>0\}}, \quad\text{with } \zeta(y)=e^y+1.
\end{gather*}

The parameters are set according to Table~\ref{tab:parameters} below.
\begin{table}[H]
\caption{Simulation parameters}
\label{tab:parameters}
\centering
\begin{tabular}{ll}
\toprule
\textbf{Parameter} & \textbf{Value}\\
\midrule
$c$ & $1$\\
$T$ & $5$ Y\\
$\eta$ & $0.5$\\
$\theta$ & $0.1$\\
$r$ & $5\%$\\
$N$ & $500$\\
$M$ & $5000$\\
\bottomrule
\end{tabular}
\end{table}

The SDEs are approximated through a classical Euler's scheme with steps length $\frac{T}{N}$, while the expectations are evaluated by means of Monte Carlo simulations with parameter $M$.

In Figure~\ref{img:dynamic} we show the dynamic strategies under EVP and VP, computed by the equations \eqref{eqn:optimal_reins_EVP} and \eqref{eqn:optimal_reins_VP}, respectively.

\begin{figure}[H]
\centering
\includegraphics[scale=0.3]{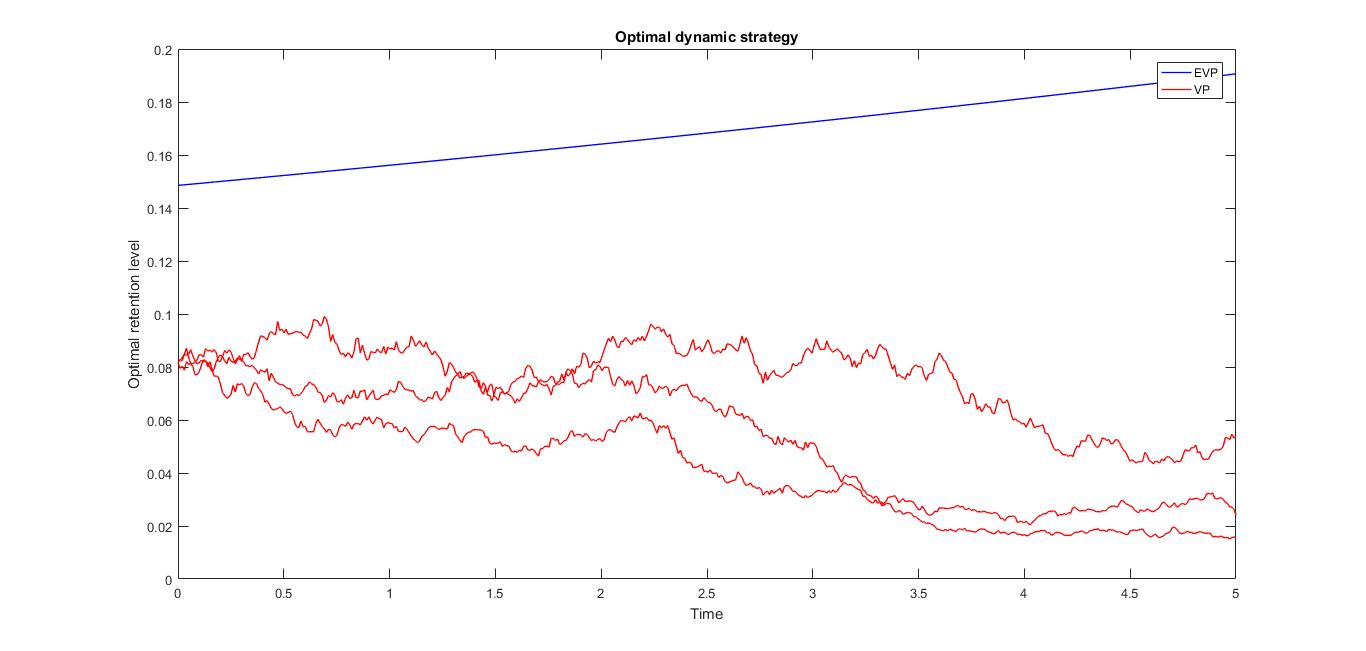}
\caption{The dynamics of the optimal strategies under EVP (red) and VP (blue).}
\label{img:dynamic}
\end{figure}

In Figure~\ref{img:eta} we start the sensitivity analysis investigating the effect of the risk aversion parameter on the optimal strategy at time $t=0$. As expected, there is an inverse relationship. Notice that for high values of $\eta$ the two strategies tend to the same level.

\begin{figure}[H]
\centering
\includegraphics[scale=0.3]{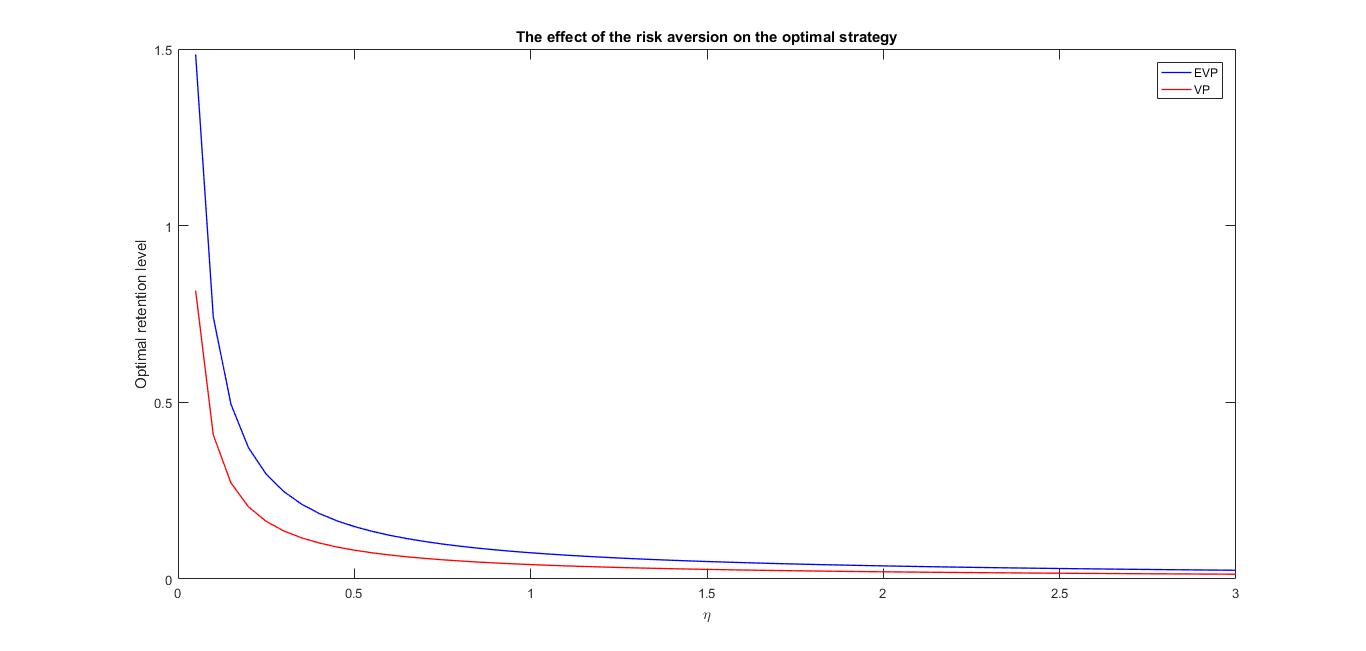}
\caption{The effect of the risk aversion on the optimal strategy under EVP (red) and VP (blue).}
\label{img:eta}
\end{figure}

Figure~\ref{img:theta} refers to the sensitivity analysis with respect to the reinsurance safety loading $\theta$. When $\theta=0$ the strategies coincide (because the premia coincide), then they diverge for increasing values of $\theta$.

\begin{figure}[H]
\centering
\includegraphics[scale=0.3]{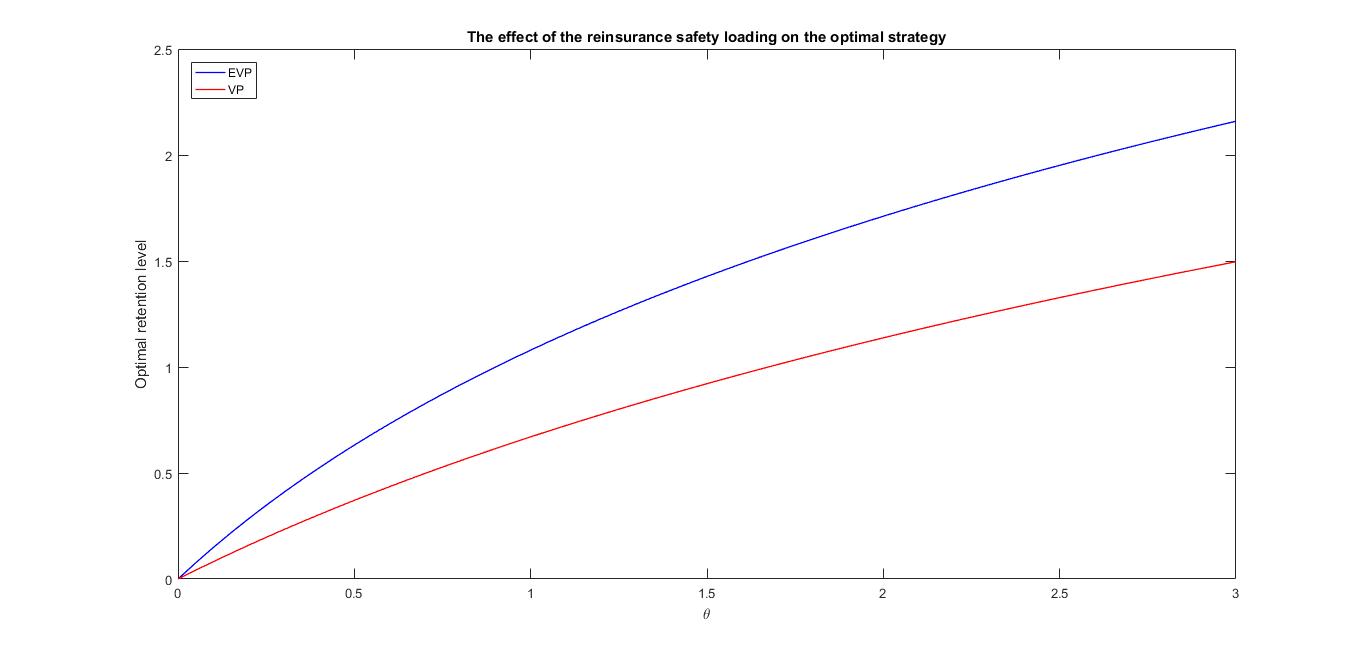}
\caption{The effect of the reinsurer's safety loading on the optimal strategy under EVP (red) and VP (blue).}
\label{img:theta}
\end{figure}

In Figure~\ref{img:r} we observe that the distance between the retention levels in the two cases is larger when $r<0$ and it decreases as long as $r$ increases. Nevertheless, even for positive values of the risk-free interest rate the distance is not negligible (see the pictures above, with $r=0.05$).

\begin{figure}[H]
\centering
\includegraphics[scale=0.3]{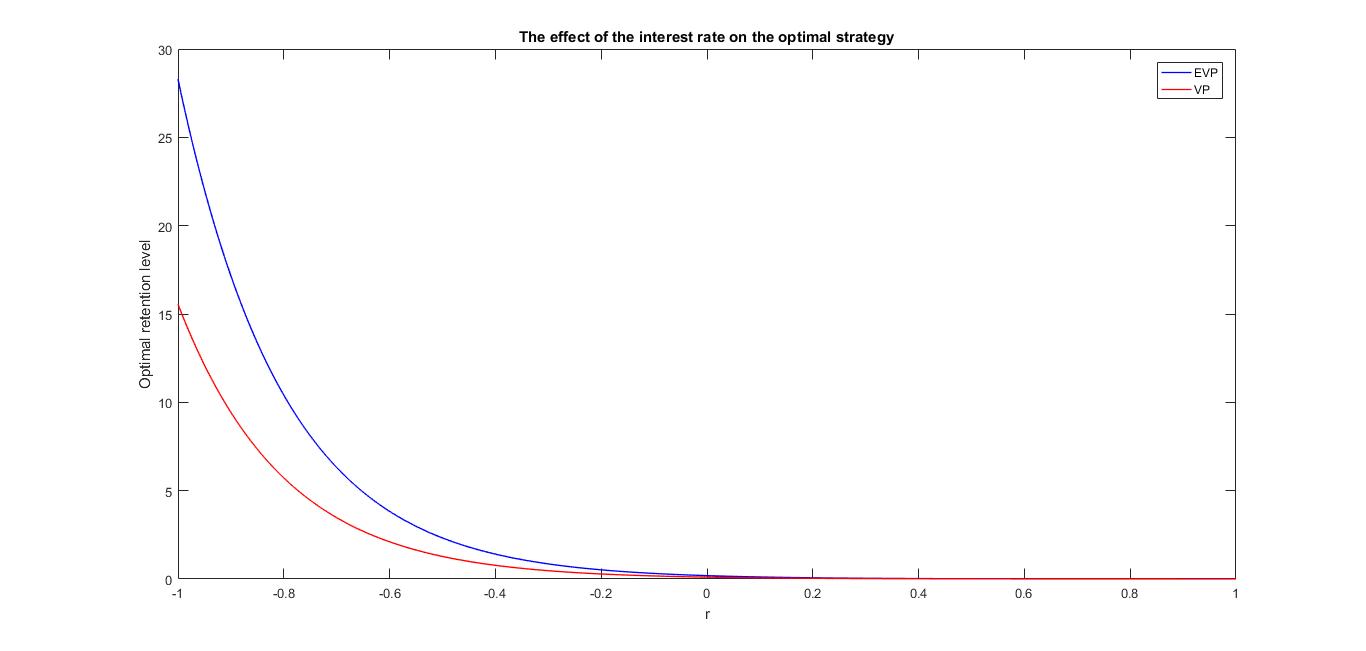}
\caption{The effect of the risk-free interest rate on the optimal strategy under EVP (red) and VP (blue).}
\label{img:r}
\end{figure}

In Figure~\ref{img:T} we study the response of the optimal strategy to variations of the time horizon. The two cases exhibit the same behavior, which is strongly influenced by the sign of the interest rate. In fact, if $r<0$ the retention level increases with the time horizon, while if $r>0$ the optimal strategy decreases with $T$.

\begin{figure}[H]
\centering
\includegraphics[scale=0.3]{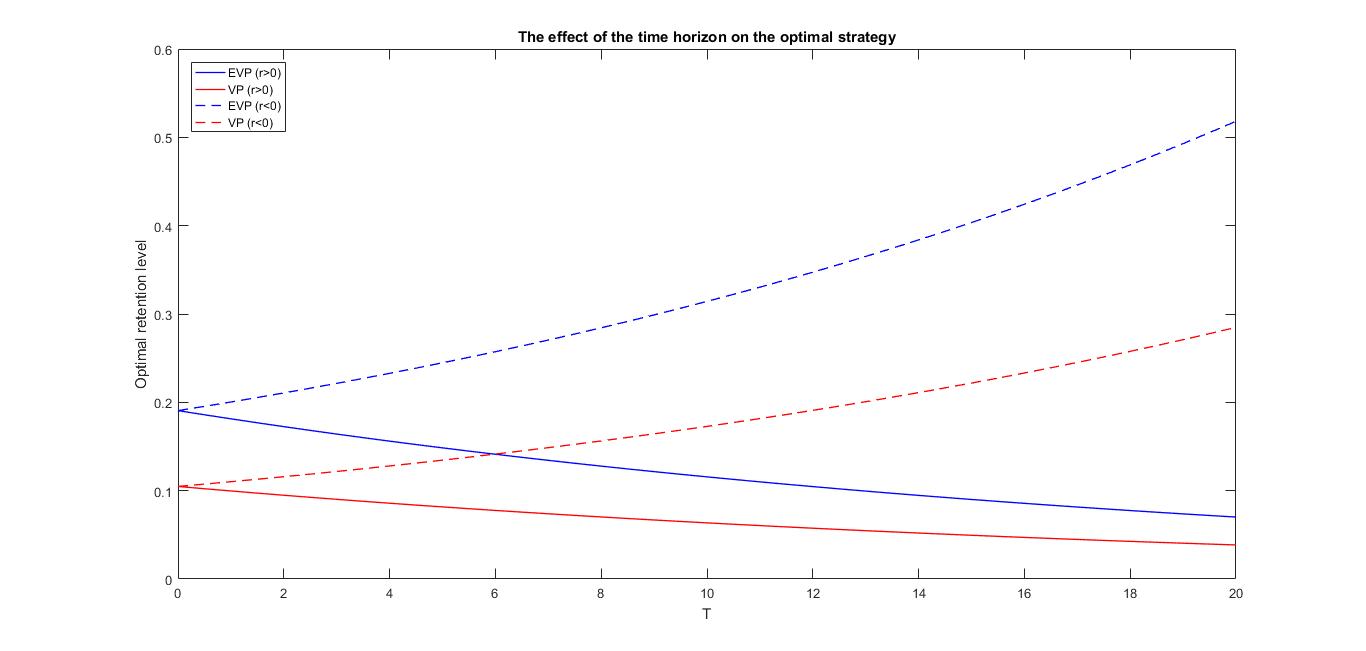}
\caption{The effect of the time horizon on the optimal strategy under EVP (red) and VP (blue).}
\label{img:T}
\end{figure}

Finally, thanks to Proposition~\ref{prop:fk} we are able to numerically approximate the value function by simulating the trajectories of $Y$. The graphical result (under VP) is shown in Figure \ref{img:vfun} below.

\begin{figure}[H]
\centering
\includegraphics[scale=0.3]{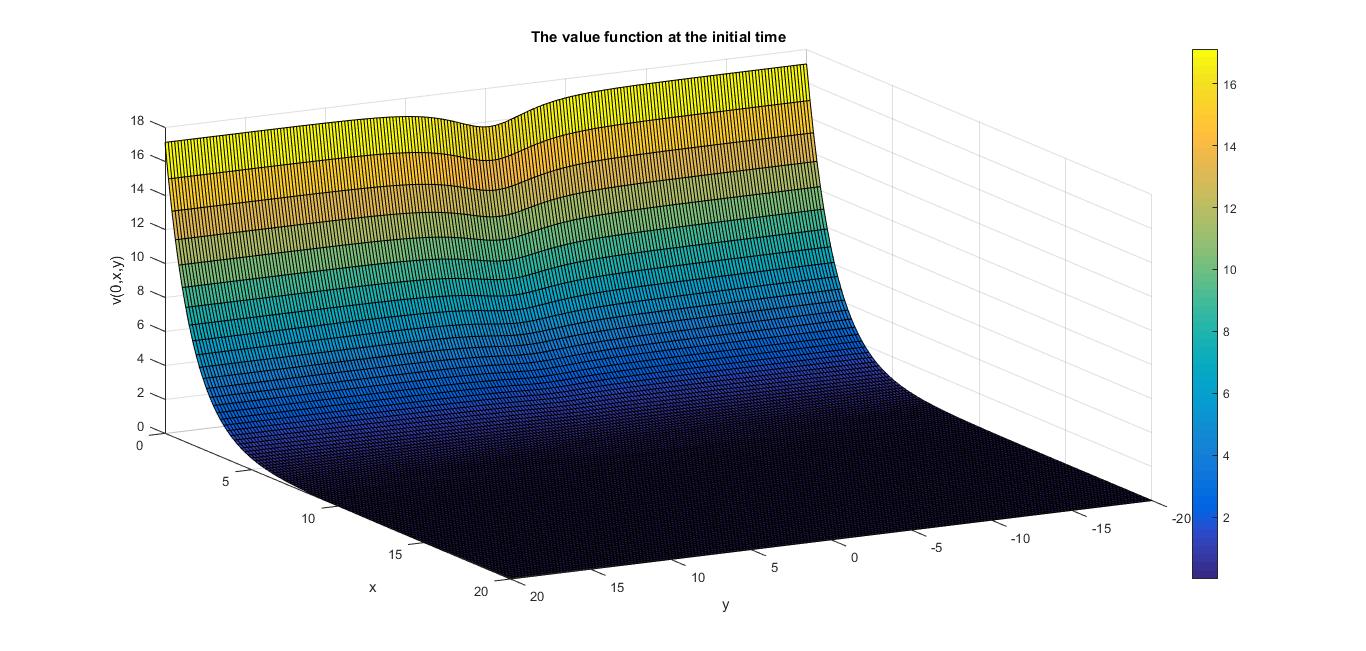}
\caption{The value function $v(0,x,y)$ at the initial time.}
\label{img:vfun}
\end{figure}


\appendix
\section{Appendix}
\label{appendix:proofs}

\begin{proof}[Proof of Proposition~\ref{prop:F_distribution}]
By equation \eqref{pp}, for any $H(t,z) = H_t \mathbbm{1}_A(z)$ with $\{H_t\}_{t\in[0,T]}$ any nonnegative $\{\mathcal{F}_t\}_{t\in[0,T]}$-predictable process  and 
$\forall A\in \mathcal{B}([0,+\infty)) $ we get 

$$\mathbb{E}\biggl[\int_0^T\int_0^{+ \infty} H_t \mathbbm{1}_A(z) m(dt,dz)\biggr]= \mathbb{E}\biggl[ \sum_{n\ge1} H_{T_n}\mathbbm{1}_{\{ Z_n \in A\}} \mathbbm{1}_{\{T_n\le T\}}\biggr] = \mathbb{E}\biggl[\int_0^T H_t \lambda(t,Y_t)  \int_A dF(z, Y_t) dt\biggr].$$

Since $H_{T_n}\mathbbm{1}_{\{T_n\le T\}}$ is an $\mathcal{F}_{T_n^-}$-measurable random variable (see Appendix 2, T4 in ~\cite{bremaud:pointproc}), denoting by 
$\mu_t(A) = \mathbb{P}[Z_n \in A\mid \mathcal{F}_{t^-} ]$  the conditional distribution of $Z_n$ given $\mathcal{F}_{t^-}$, we have that 
$$\mathbb{E}\biggl[\int_0^T\int_0^{+ \infty} H_t \mathbbm{1}_A(z) m(dt,dz)\biggr]= \mathbb{E}\biggl[ \sum_{n\ge1} H_{T_n} \mathbbm{1}_{\{T_n\le T\}} \mathbb{P}[Z_n \in A\mid \mathcal{F}_{T_n^-} ] \biggr] = $$
$$
\mathbb{E}\biggl[\int_0^T H_t  \mu_t(A) dN_t \biggr] = \mathbb{E}\biggl[\int_0^T H_t  \mu_t(A)\lambda(t,Y_t)  dt \biggr].$$
Hence the following equality holds true
$$\mathbb{E}\biggl[\int_0^T H_t  \mu_t(A)\lambda(t,Y_t)  dt \biggr] = \mathbb{E}\biggl[\int_0^T H_t \lambda(t,Y_t)  \int_A dF(z, Y_t) dt\biggr]$$
and by the arbitrariness of $\{H_t\}_{t\in[0,T]}$ and strictly positivity of $\lambda(t,Y_t)$ we finally obtain that
$$\forall A\in \mathcal{B}([0,+\infty)), \quad  \mu_t(A) =  \int_A dF(z, Y_t), \quad dt \times d\mathbb{P}-a.s..$$
\end{proof}

\section{Appendix}
\label{appendix:premium}

In this section we motivate formulas~\eqref{eqn:EVP} and~\eqref{eqn:VP}.
Let us denote by $\{C^\alpha_t\}_{t\in[0,T]}$ the reinsurer's cumulative losses at time $t$:
\[
C^\alpha_t=\int_0^t\int_0^{+\infty} (z - z\land\alpha_s)\,m(ds,dz), \qquad t\in[0,T].
\]

Recalling~\eqref{eqn:dual_projection}, by equation~\eqref{eqn:EVP} in Example~\ref{example:evp} we readily check that for any strategy $\{\alpha_t\}_{t\in[0,T]}$ under the EVP
\begin{align*}
\mathbb{E}\bigg[\int_0^tq(s,Y_s,\alpha_s)\,ds\biggr]&=
(1+\theta)\mathbb{E}\bigg[\int_0^t\int_0^{+\infty} (z - z\land\alpha_s)\,\lambda(s,Y_s)\,dF(z,Y_s)\,ds\biggr]\\
&=(1+\theta)\mathbb{E}\bigg[\int_0^t\int_0^{+\infty} (z - z\land\alpha_s)\,m(ds,dz)\biggr]\\
&=(1+\theta)\mathbb{E}[C^\alpha_t],
\end{align*}
for some safety loading $\theta>0$, i.e. for any time $t\in[0,T]$ the expected premium covers the expected losses plus an additional (proportional) term, which is the expected net income.\\

Now let us focus on Example~\ref{example:vp}. Under the VP the reinsurance premium should satisfy the following equation:
\begin{equation}
\label{eqn:vp_expected}
\mathbb{E}\bigg[\int_0^tq(s,Y_s,\alpha_s)\,ds\biggr]=\mathbb{E}[C^\alpha_t]+\theta\var[C^\alpha_t],
\end{equation}
for some safety loading $\theta>0$. We need to evaluate the variance term. Let us introduce the following stochastic process:
\[
M^\alpha_t=\int_0^t\int_0^{+\infty} (z-\alpha_s)_+\,\bigl(m(ds,dz)-\nu(ds,dz)\bigr), \qquad t\in[0,T],
\]
denoting $(x-y)_+=x-x\land y$. We have that
\begin{align*}
\var[C^\alpha_t] &= \mathbb{E}[(C^\alpha_t)^2]-\mathbb{E}[C^\alpha_t]^2\\
&= \mathbb{E}\bigl[\abs{M^\alpha_t}^2\bigr]+\mathbb{E}\biggl[\biggl(\int_0^t\int_0^{+ \infty} (z-\alpha_s)_+\,\lambda(s,Y_s) dF(z, Y_s)\,ds\biggr)^2\biggr]\\
&+2\mathbb{E}\biggl[M^\alpha_t\int_0^t\int_0^{+ \infty} (z-\alpha_s)_+\,\lambda(s,Y_s) dF(z, Y_s)\,ds\biggr]-\mathbb{E}[C^\alpha_t]^2.
\end{align*}
Denoting by $\langle M^\alpha\rangle_t$ the predictable covariance process of $M^\alpha_t$, using Remark~\ref{remark:random_measure} we finally obtain
\begin{equation*}
\begin{split}
\var[C^\alpha_t]&=\mathbb{E}[\langle M^\alpha\rangle_t]+\var\biggl[\int_0^t\int_0^{+ \infty} (z-\alpha_s)_+\,\lambda(s,Y_s) dF(z, Y_s)\,ds\biggr]\\
&=\mathbb{E}\bigg[\int_0^t\int_0^{+\infty} (z-\alpha_s)^2_+\,\lambda(s,Y_s) dF(z, Y_s)\,ds\biggr]+\var\biggl[\int_0^t\int_0^{+ \infty} (z-\alpha_s)_+\,\lambda(s,Y_s) dF(z, Y_s)\,ds\biggr].
\end{split}
\end{equation*}

Under the special case $\lambda(t,y)=\lambda(t)$ and $F(z,y)=F(z)$ (e.g. under the Cram\'er-Lundberg model), for any constant strategy $\alpha\in[0,+\infty)$ the previous equation reduces to
\[
\var[C^\alpha_t]=\mathbb{E}\bigg[\int_0^t\int_0^{+\infty} (z-\alpha)^2_+\,\lambda(s) dF(z)\,ds\biggr].
\]
Extending this formula to the model formulated in Section~\ref{section:model}, we obtain the expression~\eqref{eqn:VP}. Of course, there will be an approximation error, because in our general model the intensity and the claim size distribution depend on the stochastic factor. Nevertheless, this is a common procedure in the actuarial literature.


\newpage
\bibliographystyle{apalike}
\bibliography{bib/biblio}

\end{document}